\documentclass[12pt]{article}
\usepackage{amsfonts}
\usepackage{amssymb}
\usepackage{amsmath}
\usepackage{amsthm}
\usepackage{xcolor}
\usepackage[unicode,bookmarks,bookmarksopen,bookmarksopenlevel=2,colorlinks,linkcolor=blue,citecolor=green]{hyperref}

\newtheorem{theorem}{Theorem}

\newtheorem{definition}{Definition}

\newtheorem{proposition}{Proposition}

\newcommand{\ddx}[1]{\partial_x^{#1}}

\input{tcilatex}
\begin{document}

\title{\textbf{On a class of third-order nonlocal Hamiltonian operators}}
\author{M. Casati$^1$, E.V. Ferapontov$^1$, M.V. Pavlov$^{2}$, R.F. Vitolo%
$^{3}$ \\
$^{1}$Department of Mathematical Sciences,\\
Loughborough University,\\
Loughborough, Leicestershire, LE11 3TU, UK\\
\texttt{e.v.ferapontov@lboro.ac.uk}\\
[3mm] $^{2}$Sector of Mathematical Physics,\\
Lebedev Physical Institute of Russian Academy of Sciences,\\
Leninskij Prospekt, 53, Moscow, Russia\\
\texttt{m.v.pavlov@lboro.ac.uk}\\
[3mm] $^{3}$Department of Mathematics and Physics ``E. De Giorgi'',\\
University of Salento, Lecce, Italy\\
and INFN, Section of Lecce\\
\texttt{raffaele.vitolo@unisalento.it} }
\date{}
\maketitle

\begin{abstract}
Based on the theory of Poisson vertex algebras we calculate skew-symmetry
conditions and Jacobi identities for a class of third-order nonlocal
operators of differential-geometric type. Hamiltonian operators within this
class are defined by a Monge metric and a skew-symmetric two-form satisfying
a number of differential-geometric constraints. Complete classification
results in the 2-component and 3-component cases are obtained. \bigskip

\noindent MSC: 37K05, 37K10, 37K20, 37K25.

\bigskip

\noindent Keywords: Nonlocal Hamiltonian Operator, Monge Metric, Dirac
Reduction, Poisson Vertex Algebra.
\end{abstract}

\newpage

\tableofcontents

\allowdisplaybreaks[3]

\section{Introduction and summary of the main results}

\label{sec:intro} Third-order Hamiltonian operators of
differential-geometric type were introduced in \cite{DN2} and thoroughly
investigated in \cite{GP87, GP91, GP97, BP, Doyle}. In the so-called `flat
coordinates' $u=\{u^1, \dots, u^n\}$ these operators take the form
\begin{equation}  \label{eq:1}
A = \partial_x^{}\Big(g^{ij}\partial_x^{} + c^{ij}_k u^k_x\Big)\partial_x^{}
\end{equation}
where the coefficients $g^{ij}(u)$ and $c^{ij}_k(u)$ satisfy a system of
differential constraints coming from the skew-symmetry conditions and the
Jacobi identities. Here $i, j, k\in \{1, \dots, n\}$ where $n$ is the number
of components. Hamiltonian operators of type (\ref{eq:1}) arise in the
theory of equations of associativity of 2-dimensional topological field
theory (WDVV equations \cite{Dub}), see \cite{FN, OM98, KN1, KN2, pv}.
Projective-geometric aspects of operators (\ref{eq:1}) were studied in \cite%
{FPV, FPV1} based on their correspondence to Monge metrics and quadratic
line complexes. This has lead to complete classification results for the
number of components $n\leq 4$. The general theory of Hamiltonian systems of
conservation laws associated with operators (\ref{eq:1}) was developed in
\cite{FPV2}.

In this paper we investigate a nonlocal generalisation of ansatz (\ref{eq:1}%
),
\begin{equation}  \label{eq:32}
A = \partial_x^{}\Big(g^{ij}\partial_x^{} + c^{ij}_k
u^k_x+w^i_ku^k_x\partial_x^{-1}w^j_lu^l_x\Big)\partial_x^{},
\end{equation}
in what follows we will always assume the non-degeneracy condition $\det
g\ne 0$. Operator of type (\ref{eq:32}) appeared in the context of the
Wadati-Konno-Ishikawa (WKI) soliton hierarchy \cite{WKI, MMZ}, see Section %
\ref{sec:ex}. Although operator (\ref{eq:32}) looks analogous to first-order
nonlocal Hamiltonian operators introduced in \cite{Fer95}, the underlying
geometry is quite different.

\begin{theorem}
\label{Jacobi} Operator (\ref{eq:32}) is Hamiltonian if and only if the
following conditions are satisfied:
\begin{equation}
\begin{array}{c}
g^{ij}=g^{ji}, \\
g_{,k}^{ij}=c_{k}^{ij}+c_{k}^{ji}, \\
c_{s}^{ij}g^{sk}+c_{s}^{kj}g^{si}=0, \\
c_{s}^{ij}g^{sk}+c_{s}^{jk}g^{si}+c_{s}^{ki}g^{sj}=0, \\
g^{is}w^k_s+g^{ks}w^i_s=0, \\
g^{ks}w^j_{s,l}+g^{js}_{,l}w^k_s-c^{jk}_sw^s_l+c^{kj}_sw^s_l=0, \\
g^{ks}c^{ij}_{s,l}+c^{kj}_sg^{si}_{,l}+c^{ki}_sc^{sj}_l-c^{ik}_sc^{sj}_l+g^{ks}w^i_sw^j_l=0.%
\end{array}
\label{cond}
\end{equation}
\end{theorem}

Our proof of Theorem \ref{Jacobi} utilises the theory of Poisson vertex
algebras, see Section \ref{sec:proof}. Let us introduce an affine connection
$\nabla$ with Christoffel's symbols $\Gamma^i_{jk}=c^i_{jk}=g_{js}c^{si}_k,
\ c_{ijk}=g_{is}c^s_{jk}$ where $g_{ij}$ is the inverse matrix to $g^{ij}$.
Note that Christoffel's symbols are skew-symmetric in low indices.
Introducing the skew-symmetric 2-form $w_{ij}=g_{is}w^s_j$ we can rewrite
conditions (\ref{cond}) in the equivalent form with low indices,
\begin{equation}
\begin{array}{c}
\nabla g=0, \\
c_{ijk}+c_{ikj}=0, \\
g_{ij, k}+g_{jk, i}+g_{ki, j}=0, \\
w_{ij}+w_{ji}=0, \\
w_{ij, l}-c^s_{ij}w_{sl}=0, \\
c_{nml,k}+c^s_{ml}c_{snk}+w_{ml}w_{nk}=0.%
\end{array}
\label{cond1}
\end{equation}
The last relation implies
\begin{equation*}
R_{ijkl} = w_{il} w_{jk} - w_{ik} w_{jl}
\end{equation*}
where $R_{ijkl}=g_{is}R^s_{jkl}$ is the curvature tensor of the connection $%
\nabla$. Note that a metric $g$ satisfying the equations $g_{ij, k}+g_{jk,
i}+g_{ki, j}=0$ is the Monge metric of a quadratic line complex in $\mathbb{P%
}^n$ \cite{FPV, FPV1}. The existing classification of such metrics in
dimensions 2 and 3 leads to a complete list of 2-component (Theorem \ref{2D}
of Section \ref{sec:2}) and 3-component (Theorem \ref{3D} of Section \ref%
{sec:3}) operators (\ref{eq:32}). An important invariant of a Monge metric
is its singular variety defined by the equation $\det g=0.$ The singular
variety is an algebraic hypersurface of degree $2n-2$ \cite{Dolgachev}. For
local operators (\ref{eq:1}) the singular variety is known to be a double
hypersurface of degree $n-1$ \cite{FPV}. This is no longer the case for
nonlocal operators (\ref{eq:32}): the corresponding singular varieties are
generally irreducible.

In Theorem \ref{TD} of Section \ref{sec:Dirac} we demonstrate that $n$%
-component nonlocal operators (\ref{eq:32}) arise as Dirac reductions of $%
(n+1)$-component local operators (\ref{eq:1}) to hyperplanes in the flat
coordinates.

\medskip

\noindent\textbf{Remark.} The first three conditions (\ref{cond}) imply that
the coefficients $c^{ij}_k$ can be expressed in terms of the metric by the
formula \cite{FPV}
\begin{equation*}
c^{ij}_k=\frac{1}{3}g^{qi}g^{pj}(g_{pk, q}-g_{pq, k}).
\end{equation*}
Setting $k=j,\ l=i$ in the relation $R_{ijkl} = w_{il} w_{jk} -
w_{ik} w_{jl}$ we obtain $R_{ijji}=w_{ij}^2$, which determines $w_{ij}=\sqrt{%
R_{ijji}}$ uniquely up to a sign (which can be fixed in a consistent way
from the remaining relations up to the overall sign, $w\to -w$). Thus, to
specify Hamiltonian operator (\ref{eq:32}) it is sufficient to specify the
corresponding Monge metric $g_{ij}$.

\section{Example}

\label{sec:ex} The second flow of the $so(3)$ version of the WKI hierarchy
\cite{WKI} has the form
\begin{equation*}
\left(
\begin{array}{c}
p \\
q%
\end{array}
\right)_t= \left(
\begin{array}{c}
\frac{p_x}{(p^2+q^2+1)^{3/2}} \\
\frac{q_x}{(p^2+q^2+1)^{3/2}}%
\end{array}
\right)_{xx}.
\end{equation*}
It was demonstrated in \cite{MMZ} that this system possesses a
bi-Hamiltonian representation
\begin{equation}
\left(
\begin{array}{c}
p \\
q%
\end{array}
\right)_t=A\left(
\begin{array}{c}
\delta H/\delta p \\
\delta H/\delta q%
\end{array}
\right)=B\left(
\begin{array}{c}
\delta G/\delta p \\
\delta G/\delta q%
\end{array}
\right),
\end{equation}
with the Hamiltonians
\begin{equation*}
H=\int \sqrt{p^2+q^2+1} \ dx, ~~~~~ G=\int \frac{qp_x-pq_x}{\sqrt{p^2+q^2+1}(%
\sqrt{p^2+q^2+1}+1)}\ dx,
\end{equation*}
and the Hamiltonian operators
\begin{equation*}
A= \partial_x^{2}%
\begin{pmatrix}
\partial_x^{-1} - \tilde{q}\partial_x^{-1}\tilde{q} & \tilde{q}%
\partial_x^{-1}\tilde{p} \\
\tilde{p}\partial_x^{-1}\tilde{q} & \partial_x^{-1} - \tilde{p}%
\partial_x^{-1}\tilde{p}%
\end{pmatrix}%
\partial_x^{2}, ~~~~~ B=
\begin{pmatrix}
0 & 1 \\
-1 & 0%
\end{pmatrix}
\partial_x^{2},
\end{equation*}
where we use the notation
\begin{equation*}
\tilde{p} = \frac{p}{\sqrt{p^2+q^2+1}},\qquad \tilde{q} = \frac{q}{\sqrt{%
p^2+q^2+1}}.
\end{equation*}
The operator $A$ can be rewritten in form (\ref{eq:32}),
\begin{equation*}
A=\partial_x^{}\circ A_1\circ \partial_x^{},
\end{equation*}
where
\begin{equation*}
A_1=%
\begin{pmatrix}
1 - \tilde{q}^2 & \tilde{q}\tilde{p} \\
\tilde{p}\tilde{q} & 1 - \tilde{p}^2%
\end{pmatrix}%
\partial_x^{} +
\begin{pmatrix}
- \tilde{q}\tilde{q}_x & \tilde{p}\tilde{q}_x \\
\tilde{q}\tilde{p}_x & - \tilde{p}\tilde{p}_x%
\end{pmatrix}
+
\begin{pmatrix}
\tilde{q}_x\partial_x^{-1}\tilde{q}_x & - \tilde{q}_x\partial_x^{-1}\tilde{p}%
_x \\
- \tilde{p}_x\partial_x^{-1}\tilde{q}_x & \tilde{p}_x\partial_x^{-1}\tilde{p}%
_x%
\end{pmatrix}%
.
\end{equation*}
For the corresponding Monge metric $g_{ij}$ and the skew-symmetric 2-form $%
w_{ij}$ we obtain
\begin{equation*}
g_{ij} =
\begin{pmatrix}
q^2 + 1 & -pq \\
-pq & p^2 + 1%
\end{pmatrix}%
, ~~~ w_{ij}=\frac{1}{\sqrt{p^2+q^2+1}}%
\begin{pmatrix}
0 & 1 \\
-1 & 0%
\end{pmatrix}%
.
\end{equation*}

\section{Classification results}

The class of nonlocal operators (\ref{eq:32}) is invariant under projective
transformations of the form
\begin{equation}
\tilde u^i= \frac{l^i({\ u})}{l({\ u})}, ~~~ \tilde g= \frac{g}{l^4({\ u})},
~~~ \tilde w= \frac{w}{l^2({\ u})},  \label{group}
\end{equation}
where $l_i, l$ are linear forms in the flat coordinates. Here $g=g_{ij}$ and
$w=w_{ij}$ 
are the corresponding Monge metric and the skew-symmetric 2-form. This
symmetry reflects the reciprocal invariance of nonlocal Hamiltonian
formalism (\ref{eq:32}), thus generalising the analogous result known in the
local case \cite{FPV}. All our classification results are formulated modulo
this equivalence.

Theorem \ref{Jacobi} leads to a classification of nonlocal $n$-component
Hamiltonian operators (\ref{eq:32}) based on normal forms of Monge metrics
in dimensions $n=2$ and $n=3$ (due to the skew-symmetry of $w$ there exist
no nonlocal operators of this type for $n=1$).

\subsection{2-component case}

\label{sec:2}

\begin{theorem}
\label{2D} In the 2-component case, every Monge metric gives rise to a
Hamiltonian operator of type (\ref{eq:32}).
\end{theorem}

\begin{proof}
Every 2-component Monge metric is a quadratic form in the differentials $dp, dq$ and $ pdq-qdp$. Thus, it can be represented as
$$
g=a(pdq-qdp)^2 + 2(pdq-qdp)(bdp +cdq)+
 \alpha dp^2+2\beta dpdq+\gamma dq^2
$$
where $a, b, c, \alpha, \beta, \gamma$ are arbitrary constants. Every such metric gives rises to nonlocal operator  (\ref{eq:32}) with
$$
g_{ij}=\left(
      \begin{array}{cc}
	aq^{2}-2bq+\alpha & -apq +bp-cq+\beta\\
	-apq +bp-cq+\beta & ap^{2}+2cp+\gamma
      \end{array}
    \right), ~~~ w_{ij}=\frac{1}{\sqrt{\det g}}\begin{pmatrix}
0 & 1 \\
-1 & 0%
\end{pmatrix}.
$$
This expression can be brought to normal form using affine transformations of $p$ and $q$.

 \noindent{\bf Case $a\neq 0$}. Using translations of $p$ and $q$ we can set $b=c=0$. The rest depends on whether $\alpha \gamma-\beta^2$ is non-zero or not. In the non-zero case,
using the remaining (complex) affine freedom we can also set $a=\alpha=\gamma=1,\ \beta=0$.  This results in the metric
  \begin{equation*}
g_{ij}=\left(
      \begin{array}{cc}
	q^{2}+1 & -pq \\
	-pq & p^{2}+1
      \end{array}
    \right),
  \end{equation*}
  which corresponds to the nonlocal Hamiltonian operator $A$ from  Section
  \ref{sec:ex}. In the degenerate case $\alpha \gamma-\beta^2=0$ we can  reduce
  the metric to the form
    \begin{equation*}
g_{ij}=\left(
      \begin{array}{cc}
	q^{2}+1 & -pq \\
	-pq & p^{2}
      \end{array}
    \right),
  \end{equation*}
  which gives rise to the local Hamiltonian operator
  $$
  A=\ddx{}
  \begin{pmatrix}
    \ddx{}&  \ddx{}\frac{q}{p}\\
    \frac{q}{ p}\ddx{}& \frac{q^{2} +1}{2p^{2}}\ddx{}+\ddx{}\frac{q^{2} +1}{2p^{2}}\cr
  \end{pmatrix}
       \ddx{}.
  $$

  \noindent{\bf  Case $a=0$}. Modulo affine transformations one can always assume
  $b=1, \ c=0$ (if  $b=c=0$  we have a constant-coefficient metric  corresponding to a constant-coefficient local  operator).  Using appropriate translations of $p$ and $q$  one  can set
  $\alpha=\beta=0$. This results in the Monge metric
  \begin{equation*}
  g_{ij}= \left(
      \begin{array}{cc}
	-2q & p \\
	p & \gamma
      \end{array}%
    \right).
  \end{equation*}%
The case $\gamma \ne 0$ corresponds to  nonlocal operator (\ref{eq:32}) with the skew-symmetric 2-form $w$ defined as
$$
w_{ij}=\frac{1}{\sqrt{-2\gamma q-p^2}}\begin{pmatrix}
0 & 1 \\
-1 & 0%
\end{pmatrix}.
$$
In the case $\gamma=0$ the above metric gives rise to the local Hamiltonian operator
\begin{equation*}
A= \ddx{}\left(
\begin{array}{cc}
0 & \displaystyle\ddx{}\frac{1}{p} \\
\displaystyle\frac{1}{p}\ddx{} & \displaystyle
\frac{q}{p^{2}}\ddx{}+\ddx{}\frac{q}{p^{2}}
\end{array}%
\right) \ddx{},
\end{equation*}
which appeared as a Hamiltonian structure of Monge-Amp\`ere equations \cite{OM98}, see also \cite{FPV}.
\end{proof}



\subsection{3-component case}

\label{sec:3}

Every 3-component Monge metric $g$ can be written as a quadratic form in the
6 differentials $du^i, \ u^idu^j-u^jdu^i, \ i, j=1, 2, 3$.
Let $Q$ denote the $6\times 6$ symmetric matrix of this quadratic form. Let $%
P$ denote the $6\times 6$ symmetric matrix corresponding to the quadratic
Pl\"ucker relation,
\begin{equation*}
du^1(u^2du^3-u^3du^2)+du^2(u^3du^1-u^1du^3)+du^3(u^1du^2-u^2du^1)=0.
\end{equation*}
Monge metrics are classified by their Segre types, that is, Jordan normal
forms of the operator $Q P^{-1}$. In what follows we use the standard
notation: thus, Segre type $[123]$ indicates that the operator $Q P^{-1}$
has three Jordan blocks of sizes $1\times 1, 2\times 2$ and $3\times 3$,
respectively. Additional round brackets indicate coincidences among the
eigenvalues of these blocks: thus, $[(12)3]$ indicates that the eigenvalue $%
\lambda_1$ of the first Jordan block coincides with the eigenvalue $%
\lambda_2 $ of the second one, etc. We refer to \cite{Jess, fmoss} for the
list of normal forms of 3-component Monge metrics. All classification
results are formulated modulo projective equivalence (\ref{group}). In what
follows we only present the Monge metric $g_{ij}$ and the skew-symmetric
2-form $w_{ij}$ (which uniquely specify the corresponding nonlocal operator (%
\ref{eq:32}); note that the 2-form $w$ is defined up to an overall sign). The
Theorem below provides a complete description of 3-component nonlocal
operators (\ref{eq:32}) by going through the list of all Segre types and
indicating particular allowed subcases that give rise to nonlocal operators.
These are singled out by  conditions (\ref{cond1}). In each case we
explicitly state the equation of the singular surface, $\det g=0$, which is
a quartic in $\mathbb{P}^3$ (possibly, reducible). It turns out that this
quartic degenerates into a double quadric if and only if the operator is
local.

\begin{theorem}
\label{3D} Modulo (complex) projective transformations (\ref{group}) any
nonlocal Hamiltonian operator \eqref{eq:32} can be reduced to one of the
following normal forms in the Segre classification:

\begin{enumerate}
\item \textbf{Segre type $[114]$.} Here the only allowed subcase is $[(114)]$
which corresponds to the local operator defined by the metric $g^{(4)}$ from
\cite{FPV}:
\begin{equation*}
g_{ij}^{(4)}=
\begin{pmatrix}
-2u^2 & u^1 & 0 \\
u^1 & 0 & 0 \\
0 & 0 & 1%
\end{pmatrix}%
, ~~~ w_{ij}=0.
\end{equation*}
We have $\det g^{(4)} = - (u^1)^2$, the singular surface is a pair of double
planes (one of them at infinity).

\item \textbf{Segre type $[123]$.} Here the only allowed subcase is $[(123)]$
which corresponds to the local operator defined by the metric $g^{(5)}$ from
\cite{FPV}:
\begin{equation*}
g^{(5)}_{ij}=%
\begin{pmatrix}
-2u^2 & u^1 & 1 \\
u^1 & 1 & 0 \\
1 & 0 & 0%
\end{pmatrix}%
, ~~~ w_{ij}=0.
\end{equation*}
We have $\det g^{(5)} = -1$, hence the singular surface is the quadruple
plane at infinity.

\item \textbf{Segre type $[222]$.} Here the only allowed subcase is $[(222)]$
which corresponds to the local operator defined by the metric $g^{(6)}$ from
\cite{FPV}:
\begin{equation*}
g^{(6)}_{ij} =
\begin{pmatrix}
1 & 0 & 0 \\
0 & 1 & 0 \\
0 & 0 & 1%
\end{pmatrix}%
, ~~~ w_{ij}=0.
\end{equation*}
We have $\det g^{(6)} = 1 $, the singular surface is the quadruple plane at
infinity.

\item \textbf{Segre type $[15]$.} Here the only allowed subcase is $[(15)]$
which gives rise to the nonlocal operator with the following metric $g$ and
2-form $w$:
\begin{gather*}
g_{ij}=%
\begin{pmatrix}
0 & 1 & 2 u^{3} \\
1 & -2 u^{3} & u^{2} \\
2 u^{3} & u^{2} & -4 u^{1}%
\end{pmatrix}%
, \\
w_{12} = 0 ,\quad w_{31}=-\frac{1}{ \sqrt{u^{1} +u^{2} u^{3} +2 (u^{3})^{3}}}%
,\quad w_{23}=\frac{u^{3}}{ \sqrt{u^{1} +u^{2} u^{3} +2 (u^{3})^{3}}}.
\end{gather*}
We have $\det g = 4 u^{1} +4u^{2} u^{3} +8 (u^{3})^{3}$, the singular
surface is a Cayley's ruled cubic and the plane at infinity.

\item \textbf{Segre type $[24]$.} Here the only allowed subcase is $[(24)]$,
which further splits into two projectively dual subcases. The first subcase
gives rise to the nonlocal operator with the following metric $g$ and 2-form
$w$:
\begin{gather*}
g_{ij}=
\begin{pmatrix}
1 & 0 & u^{3} \\
0 & 1 & 0 \\
u^{3} & 0 & -2 u^{1}%
\end{pmatrix}%
, \\
w_{12}=w_{23}=0,\qquad w_{31}=\frac{1}{\sqrt{2 u^1 + (u^3)^2 }}.
\end{gather*}
We have $\det g = -2 u^1 - (u^3)^2$, the singular surface is a quadratic
cone and a double plane at infinity. The second subcase corresponds to
\begin{gather*}
g_{ij}=%
\begin{pmatrix}
0 & 1 & 0 \\
1 & (u^3)^2 & - u^2 u^3 \\
0 & -u^2 u^3 & 1+(u^2)^2%
\end{pmatrix}%
, \\
w_{12}=w_{31}=0,\qquad w_{23}=\frac{1}{\sqrt{1+(u^2)^2}}.
\end{gather*}
We have $\det g = -1 - (u^2)^2$, the singular surface is a pair of planes
and the double plane at infinity.

\item \textbf{Segre type $[33]$.} Here the only allowed subcase is $[(33)]$
which gives rise to the nonlocal operator with the following metric $g$ and
2-form $w$:
\begin{gather*}
g_{ij}=%
\begin{pmatrix}
0 & 1 & 1 \\
1 & -2u^3 & u^2+ u^3 \\
1 & u^2+ u^3 & -2 u^2%
\end{pmatrix}%
, \\
w_{12}=w_{31}=0,\qquad w_{23}=\frac{2}{\sqrt{4u^2+4u^3}}.
\end{gather*}
We have $\det g =4(u^2+u^3)$, the singular surface is a plane and another
triple plane at infinity.

\item \textbf{Segre type $[6]$.} This case does not correspond to any
Hamiltonian operator.

\item \textbf{Segre type $[1122]$.} There are 3 allowed subcases:

\begin{itemize}
\item subcase [(11)22] with the additional constraint $2\lambda_1=\lambda_3+%
\lambda_4$. This gives rise to the nonlocal operator with the following
metric $g$ and 2-form $w$:
\begin{gather*}
g_{ij}=%
\begin{pmatrix}
1 & -2\lambda u^3 & \lambda u^2 \\
-2\lambda u^3 & 4 & \lambda u^1 \\
\lambda u^2 & \lambda u^1 & 0%
\end{pmatrix}%
, \\
w_{12}=0,\qquad w_{23}=\frac{\lambda^2u^1}{\sqrt{\det g}},\qquad w_{31}=%
\frac{\lambda^2u^2}{\sqrt{\det g}},
\end{gather*}
where $\lambda=\lambda_3-\lambda_4$. We have $\det g
=-\lambda^2(u^1)^2-4\lambda^2(u^2)^2-4\lambda^3u^1u^2u^3$, the singular
surface is a cubic and the plane at infinity.

\item subcase [11(22)] with the additional constraint $\lambda_1+\lambda_2=2%
\lambda_3$. This gives rise to the nonlocal operator with the following
metric $g$ and 2-form $w$:
\begin{gather*}
g_{ij}=%
\begin{pmatrix}
1+\lambda (u^2)^2 & -\lambda u^1 u^2 & 0 \\
-\lambda u^1 u^2 & 4+\lambda(u^1)^2 & 0 \\
0 & 0 & \lambda%
\end{pmatrix}%
, \\
w_{12}=\frac{2\lambda}{\sqrt{\det g}},\qquad w_{31}= w_{23}=0,
\end{gather*}
where $\lambda=\lambda_1-\lambda_2$. We have $\det g
=4\lambda+\lambda^2(u^1)^2+4\lambda^2(u^2)^2$, the singular surface is a
quadratic cone and the double plane at infinity.

\item subcase [1(12)2], with the additional constraint $3\lambda_2=2%
\lambda_4+\lambda_1$. This gives rise to the nonlocal operator with the
following metric $g$ and 2-form $w$:
\begin{gather*}
g_{ij}=%
\begin{pmatrix}
1+\lambda (u^2)^2 & \lambda(- u^1 u^2+u^3) & \lambda u^2 \\
\lambda(- u^1 u^2+u^3) & 4+\lambda(u^1)^2 & -2\lambda u^1 \\
\lambda u^2 & -2\lambda u^1 & \lambda%
\end{pmatrix}%
, \\
w_{12}=\frac{\lambda\sqrt{\lambda}u^1}{\sqrt{-\det g}},\qquad
w_{23}=0,\qquad w_{31}=-\frac{\lambda\sqrt{\lambda}}{\sqrt{-\det g}},
\end{gather*}
where $\lambda=\lambda_1-\lambda_2$. We have $\det g
=4\lambda-3\lambda^2(u^1)^2-4\lambda^3(u^1u^2)^2+\lambda^3(u^3-u^1u^2)^2$,
the singular surface is an irreducible quartic.
\end{itemize}

\item \textbf{Segre type $[1113]$.} There are 2 essentially different
allowed subcases:

\begin{itemize}
\item subcase [11(13)], with the additional constraint $2\lambda_3=%
\lambda_1+\lambda_2$. This gives rise to the nonlocal operator with the
following metric $g$ and 2-form $w$:
\begin{gather*}
g_{ij}=%
\begin{pmatrix}
2u^3+\lambda (u^2)^2 & -1-\lambda u^1 u^2 & -u^1 \\
-1-\lambda u^1 u^2 & \lambda(u^1)^2 & 0 \\
-u^1 & 0 & \lambda%
\end{pmatrix}%
, \\
w_{12}=\frac{\lambda}{\sqrt{-\det g}},\qquad w_{23}=0,\qquad w_{31}=\frac{%
\lambda u^1}{\sqrt{-\det g}},
\end{gather*}
where $\lambda=\lambda_1-\lambda_2$. We have $\det g =-\lambda-2\lambda^2
u^1u^2+2\lambda^2u^3(u^1)^2-\lambda(u^1)^4$, the singular surface is an
irreducible quartic.

\item subcase [(11)13], with the additional constraint $4\lambda_1=%
\lambda_3+3\lambda_4$. This gives rise to the nonlocal operator with the
following metric $g$ and 2-form $w$:
\begin{gather*}
g_{ij}=%
\begin{pmatrix}
2u^3-2\lambda(u^3)^2 & -1-2\lambda u^3 & -u^1+\lambda u^2+2\lambda u^1 u^3
\\
-1-2\lambda u^3 & -2\lambda & \lambda u^1 \\
 -u^1+\lambda u^2+2\lambda u^1 u^3 & \lambda u^1 & -2\lambda(u^1)^2%
\end{pmatrix}%
, \\
\omega_{12}=0,\qquad \omega_{23}=\frac{\sqrt{3}\lambda^2u^1}{\sqrt{-\det g}}%
,\qquad \omega_{31}=\frac{\sqrt{3}\lambda^2u^2}{\sqrt{-\det g}},
\end{gather*}
where $\lambda=2(\lambda_1-\lambda_4)$. We have $\det g =
2  \lambda  (\lambda^{2}  (u^{1})^{2}  (u^{3})^{2}
  +2  \lambda^{2}  u^{1}  u^{2}  u^{3}
  +\lambda^{2}  (u^{2})^{2}
  +3  \lambda  (u^{1})^{2}  u^{3}
  -3  \lambda  u^{1}  u^{2}
  +3  (u^{1})^{2}
)$, the singular surface is an
irreducible quartic.
\end{itemize}

\item \textbf{Segre type $[11112]$.} There are 2 essentially different
allowed subcases:

\begin{itemize}
\item subcase $[111(12)]$, with the additional constraint $%
\lambda_1+\lambda_2+\lambda_3=3\lambda_4$. This gives rise to the nonlocal
operator with the following metric $g$ and 2-form $w$:
\begin{gather*}
g_{ij}=%
\begin{pmatrix}
1+\lambda(u^2)^2+\mu(u^3)^2 & -\lambda(u^3+u^1 u^2) & -\lambda u^2-\mu u^1
u^3 \\
-\lambda(u^3+u^1 u^2) & \lambda (u^1)^2+\mu & 2\lambda u^1 \\
-\lambda u^2-\mu u^1 u^3 & 2\lambda u^1 & \lambda+\mu(u^1)^2%
\end{pmatrix}%
, \\
w_{12}=\sqrt{\frac{\lambda(\mu^2-\lambda^2)}{\det g}}u^1,\qquad
w_{23}=0,\qquad w_{31}=\sqrt{\frac{\lambda(\mu^2-\lambda^2)}{\det g}},
\end{gather*}
where $\lambda=\lambda_1-\lambda_4$ and $\mu=\lambda_3-\lambda_2$. We have \begin{multline*}
  \det g = -\lambda^{3}  (u^{1} u^{2})^{2}
  +2  \lambda^{3} u^{1} u^{2} u^{3}
  -\lambda^{3}  (u^{3})^{2}-3  \lambda^{2}  (u^{1})^{2}
  \\
 +\lambda  \mu ^{2}  (u^{1} u^{2})^{2}
-2  \lambda  \mu ^{2}  u^{1}  u^{2}  u^{3}+\lambda  \mu ^{2}  (u^{3})^{2}
+\lambda  \mu   (u^{1})^{4}+\lambda  \mu +\mu ^{2}  (u^{1})^{2},
\end{multline*}
the
singular surface is an irreducible quartic.
Note that the additional
constraint $\lambda^2=\mu^2$ leads to a local operator.

\item subcase $[(11)112]$, with the additional constraint $%
4\lambda_1=\lambda_3+\lambda_4+2\lambda_5$. This gives rise to the nonlocal
operator with the following metric $g$ and 2-form $w$:
\begin{gather*}
g_{ij}=%
\begin{pmatrix}
1+\mu(u^3)^2 & -2\beta u^3 & \beta u^2-\mu u^1 u^3 \\
-2\beta u^3 & \mu & \beta u^1 \\
\beta u^2-\mu u^1 u^3 & \beta u^1 & \mu(u^1)^2%
\end{pmatrix}%
, \\
w_{12}=0,\qquad w_{23}=\sqrt{\frac{\beta^2(\beta^2-\mu^2)}{\det g}}%
u^1,\qquad w_{31}=\sqrt{\frac{\beta^2(\beta^2-\mu^2)}{\det g}}u^2,
\end{gather*}
where $\mu=\lambda_3-\lambda_4$ and $\beta=2(\lambda_1-\lambda_5)$. We have 
\begin{multline*}
  \det g = 
2  \lambda  (\lambda^{2}  (u^{1} u^{3})^{2}
  +2  \lambda^{2}  u^{1}  u^{2}  u^{3}  +\lambda^{2}  (u^{2})^{2}
  +3  \lambda  (u^{1})^{2}  u^{3}  -3  \lambda  u^{1}  u^{2}  +3 (u^{1})^{2}),
\end{multline*}
the
singular surface is an irreducible quartic. Note that the additional
constraint $\beta^2=\mu^2$ leads to a local operator.
\end{itemize}


\item \textbf{Segre type $[111111]$.} Up to relabelling of the eigenvalues,
there is essentially only one allowed subcase, namely [1111(11)], with the
additional constraint $\lambda_1+\lambda_2+\lambda_3+\lambda_4=4\lambda_5$.
This gives rise to the nonlocal operator with the following metric $g$ and
2-form $w$:
\begin{gather*}
g_{ij}=%
\begin{pmatrix}
a_3(u^2)^2+a_2(u^3)^2 & -a_3 u^1u^2+\alpha u^3 & -a_2 u^1 u^3+\alpha u^2 \\
-a_3 u^1u^2+\alpha u^3 & a_2+a_3(u^1)^2 & -2\alpha u^1 \\
-a_2 u^1 u^3+\alpha u^2 & -2\alpha u^1 & a_3+a_2(u^1)^2%
\end{pmatrix}%
, \\
w_{12}=\sqrt{\frac{(a_2^2-\alpha^2)(a_3^2-\alpha^2)}{\det g}}u^3,\qquad
w_{23}=0,\qquad w_{31}=\sqrt{\frac{(a_2^2-\alpha^2)(a_3^2-\alpha^2)}{\det g}}%
u^2,
\end{gather*}
where $a_2=\lambda_3-\lambda_4$, $a_3=\lambda_1-\lambda_2$ and $%
\alpha=2\lambda_5-\lambda_3-\lambda_4$. We have 
\begin{multline*}
  \det g = a_{2}^{2}  a_{3}  (u^{1} u^{2})^{2}
+a_{2}^{2}  a_{3}  (u^{3})^{2}
+2  a_{2}^{2}  \alpha  u^{1}  u^{2}  u^{3}
\\
+a_{2}  a_{3}^{2}  (u^{1}u^{3})^{2}
+a_{2}  a_{3}^{2}  (u^{2})^{2}
-a_{2}  \alpha^{2}  (u^{1} u^{3})^{2}
-a_{2}  \alpha^{2}  (u^{2})^{2}
\\
+2  a_{3}^{2}  \alpha  u^{1} u^{2} u^{3}
-a_{3}  \alpha^{2}  (u^{1} u^{2})^{2}
-a_{3}  \alpha^{2}  (u^{3})^{2}
-4  \alpha^{3}  u^{1}  u^{2}  u^{3},
\end{multline*}
the singular surface is an
irreducible quartic.
Note that the additional constraints $%
\alpha_2^2=\alpha^2$ or $\alpha_2^2=\alpha^2$ lead to local operators.
\end{enumerate}
\end{theorem}

\section{Nonlocal operators via Dirac reduction}

\label{sec:Dirac}

Let us consider an $(n+1)$-component local third-order Hamiltonian operator
\begin{equation}
A ^{IJ}=\partial_x^{}\Big(G^{IJ}(u)\partial_x^{} + C^{IJ}_K(u) u^K_x\Big)%
\partial_x^{},  \label{n+1}
\end{equation}
represented in the flat coordinates $u^0, \dots, u^n$.
In this section we
will calculate the Dirac reduction of this operator to a hyperplane in the $%
u $-coordinates. Without any loss of generality one can assume that this
hyperplane is given by the equation $u^0=0$. In what follows we use the
following convention for the small and capital indices: $i, j, k\in \{1,
\dots, n\}, \ I, J, K\in \{0, \dots, n\}$.

\begin{theorem}
\label{TD} Dirac reduction of local Hamiltonian operator (\ref{n+1}) to the
hyperplane $u^0=0$ is given by nonlocal operator (\ref{eq:32}) where $%
g_{ij}=G_{ij},\ c_{ijk}=C_{ijk},\ w_{ij}=\frac{C^0_{ij}}{\sqrt{G^{00}}}$.
\end{theorem}

\begin{proof} We find it more convenient to work in potential coordinates $b^{K}$ defined
as $u^{K}=b_{x}^{K}$. In these coordinates operator (\ref{n+1}) takes
first-order form,
\begin{equation*}
A^{IJ}=-G^{IJ}(b_{x})\partial _{x}-C_{K}^{IJ}(b_{x})b_{xx}^{K}.
\end{equation*}%
Its Dirac reduction $\tilde{A}$ to the hyperplane $b^{0}=0$ is defined by the
formula
\begin{equation*}
\tilde{A}^{ij}=A^{ij}-A^{i0}(A^{00})^{-1}A^{0j}.
\end{equation*}%
Since $A^{00}=-G^{00}\partial _{x}-C_{K}^{00}b_{xx}^{K}$ and $%
G_{,K}^{00}=2C_{K}^{00}$ we obtain
\begin{equation*}
A^{00}=-G^{00}\partial _{x}-\frac{1}{2}G_{,K}^{00}b_{xx}^{K}=-G^{00}\partial
_{x}-\frac{1}{2}(G^{00})_{x}=-\sqrt{G^{00}}\partial _{x}\sqrt{G^{00}},
\end{equation*}%
so that
\begin{equation*}
(A^{00})^{-1}=-\frac{1}{\sqrt{G^{00}}}\partial _{x}^{-1}\frac{1}{\sqrt{G^{00}}}.
\end{equation*}%
Thus,
\begin{equation*}
\tilde{A}^{ij}=A^{ij}+A^{i0}\frac{1}{\sqrt{G^{00}}}\partial _{x}^{-1}\frac{1%
}{\sqrt{G^{00}}}A^{0j},
\end{equation*}%
where%
\begin{equation*}
A^{ij}=-G^{ij}\partial _{x}-C_{k}^{ij}b_{xx}^{k},
\end{equation*}%
note that $b^{0}=0$. Explicitly, this gives
\begin{equation*}
\tilde{A}^{ij}=-G^{ij}\partial _{x}-C_{k}^{ij}b_{xx}^{k}+(G^{i0}\partial
_{x}+C_{k}^{i0}b_{xx}^{k})\frac{1}{\sqrt{G^{00}}}\partial _{x}^{-1}\frac{1}{%
\sqrt{G^{00}}}(G^{0j}\partial _{x}+C_{k}^{0j}b_{xx}^{k}).
\end{equation*}%
This expression can be rewritten in the form 
\begin{equation*}
\tilde{A}^{ij}=-{g}^{ij}\partial _{x}-{c}%
_{k}^{ij}b_{xx}^{k}-w_{k}^{i}b_{xx}^{k}\partial _{x}^{-1}w_{m}^{j}b_{xx}^{m}
\end{equation*}%
which reduces to (\ref{eq:32}) in the original variables $u^i=b^i_x$. Here%
\begin{equation*}
g^{ij}=G^{ij}-\frac{G^{i0}G^{0j}}{G^{00}},\text{ \ }{c}_{k}^{ij}=C_{k}^{ij}-%
\frac{G^{i0}C_{k}^{0j}+C_{k}^{i0}G^{0j}}{G^{00}}+\frac{%
G^{i0}G^{0j}G_{,k}^{00}}{2(G^{00})^{2}},
\end{equation*}%
\begin{equation*}
w_{k}^{i}=\frac{1}{\sqrt{G^{00}}}C_{k}^{i0}-\frac{1}{2}\frac{%
G^{i0}G_{,k}^{00}}{(G^{00})^{3/2}}.
\end{equation*}%
The formula for $g^{ij}$ precisely means that $g_{ij}=G_{ij}$. The formula for $c_{ijk}$
follows from the fact the $c$ is determined by $g$ (see Remark at the end of
Section \ref{sec:intro}). It remains to prove that the expression for $w^i_k$ is equivalent to the formula $w_{ij}=\frac{%
C_{ij}^{0}}{\sqrt{G^{00}}}$.
Since
\begin{equation*}
C_{k}^{i0}=G^{i0}C_{0k}^{0}+G^{im}C_{mk}^{0}
\end{equation*}%
we obtain
\begin{equation*}
w_{k}^{i}=\frac{1}{\sqrt{G^{00}}}[G^{i0}C_{0k}^{0}+G^{im}C_{mk}^{0}]-\frac{1%
}{2}\frac{G^{i0}}{(G^{00})^{3/2}}G_{,k}^{00}.
\end{equation*}%
Taking into account%
\begin{equation*}
G^{ik}=g^{ik}+\frac{G^{i0}G^{0k}}{G^{00}}
\end{equation*}%
this gives
\begin{equation*}
w_{k}^{i}=\frac{1}{(G^{00})^{3/2}}\left[
G^{00}G^{i0}C_{0k}^{0}+(G^{00}g^{im}+G^{i0}G^{0m})C_{mk}^{0}-\frac{1}{2}%
G^{i0}G_{,k}^{00}\right] .
\end{equation*}%
Using the derivative of the inverse matrix,
\begin{equation*}
G_{,K}^{IJ}=-G^{IP}G_{PQ,K}G^{QJ},
\end{equation*}%
we obtain%
\begin{equation*}
G_{,k}^{00}=-(G^{00})^{2}G_{00,k}-2G^{00}G^{0m}G_{0m,k}-G^{0m}G^{0s}g_{ms,k}.
\end{equation*}%
Taking into account%
\begin{equation*}
C_{0k}^{0}=G^{00}C_{00k}+G^{0m}C_{m0k},\text{ \ }%
C_{ik}^{0}=G^{00}C_{0ik}+G^{0s}C_{sik},
\end{equation*}%
on simplification we obtain%
\begin{equation*}
w_{k}^{i}=\frac{g^{im}C^0_{mk}}{\sqrt{G^{00}}},
\end{equation*}%
which is equivalent to the required formula $w_{ij}=\frac{%
C_{ij}^{0}}{\sqrt{G^{00}}}$.
\end{proof}


\section{Skew-symmetry conditions and Jacobi identities: proof of Theorem
\protect\ref{Jacobi}}

\label{sec:proof}

The standard way to calculate skew-symmetry conditions and Jacobi identities
is based on the Gelfand-Dorfman approach \cite{GD}. In this Section we
utilise an alternative technique based on the theory of Poisson vertex algebras \cite{PV1, PV2}
which gives a completely algebraic approach to local and nonlocal
Hamiltonian operators. This is achieved by considering the differential
algebra corresponding, in the theory of the formal calculus of variations,
to the densities of local functionals -- usually, the space of differential
polynomials or some extension thereof. More precisely, any Poisson vertex
algebra defines a Poisson bracket on the space of local functionals and an
action of the space of local functionals on the space of their densities:
such objects are, equivalently, defined by Hamiltonian operators.
Conversely, a Hamiltonian operator on the space of densities, either
differential or pseudodifferential (under some additional technical
hypotheses), defines a Poisson vertex algebra.

\begin{definition}
A (nonlocal) Poisson vertex algebra (PVA) is a differential algebra $(%
\mathcal{A},\partial)$ endowed with a derivation $\partial$ and a bilinear
operation $\{\cdot_\lambda\cdot\}\colon\mathcal{A}\otimes\mathcal{A}\to%
\mathbb{R}((\lambda^{-1}))\otimes\mathcal{A}$ called a (nonlocal) \emph{$%
\lambda$ bracket}, satisfying the following set of properties:

\begin{enumerate}
\item $\{\partial f_\lambda g\}=-\lambda\{f_\lambda g\}$ (left
sesquilinearity),

\item $\{f_\lambda \partial g\}=(\lambda+\partial)\{f_\lambda g\}$ (right
sesquilinearity),

\item $\{f_\lambda gh\}=\{f_\lambda g\}h +\{f_\lambda h\}g$ (left Leibnitz
property),

\item $\{fg_\lambda h\}=\{f_{\lambda+\partial} h\}g+\{g_{\lambda+\partial}
h\} f$ (right Leibnitz property),

\item $\{g_\lambda f\}=-{}_\to\{f_{-\lambda-g}g\}$ (PVA skew-symmetry),

\item $\{f_\lambda\{g_\mu h\}\}-\{g_\mu\{f_\lambda h\}\}=\{\{f_\lambda
g\}_{\lambda+\mu} h\}$ (PVA-Jacobi identity).
\end{enumerate}
\end{definition}

Let us denote
\begin{equation*}
\{f_\lambda g\}=\sum_{s\leq S} C_s(f,g)\lambda^s.
\end{equation*}
The expansion of the bracket in $\lambda$ is bounded by $0\leq s\leq S$ for
\emph{local} PVAs and is not bounded from below for \emph{nonlocal} PVAs.
Using the expansion, the expressions on the RHS of Property 4 are to be
understood as $\{f_{\lambda+\partial}g\}h=\sum C_s(f,g)(\lambda+\partial)^s
h=\sum_s\sum_t C_s\binom{s}{t} \partial^t h\lambda^{s-t}$, while the RHS of
Property 5 reads ${}_\to\{f_{-\lambda-\partial}g\}=\sum_s(-\lambda-%
\partial)^sC_s(f,g)$.

For the case of nonlocal PVAs, it should be noted that the three terms of
PVA-Jacobi identity do not necessarily belong to the same space, because of
the double infinite expansion of the brackets (in terms of $(\lambda,\mu)$, $%
(\mu,\lambda)$ and $(\lambda,\lambda+\mu)$, respectively). A bracket is said
to be \emph{admissible} if all the three terms can be (not uniquely)
expanded as
\begin{equation*}
\{f_\lambda\{g_\mu h\}\}=\sum_{m\leq M}\sum_{n\leq N}\sum_{p\leq
0}a_{m,n,p}\lambda^m\mu^n(\lambda+\mu)^p,
\end{equation*}
and only admissible brackets can define a nonlocal PVA. We denote the space
where the PVA-Jacobi identity of admissible brackets takes values as $%
V_{\lambda,\mu}$. This space can be decomposed by the total degree $d$ in $%
(\lambda,\mu,\lambda+\mu)$; finally, elements of each homogeneous component $%
V^{(d)}_{\lambda,\mu}$ can be \emph{uniquely} expressed in the basis \cite%
{PV2}
\begin{align*}
\lambda^i\mu^{d-i}& & i&\in \mathbb{Z}, \\
\lambda^{d+i}(\lambda+\mu)^{-i}& & i&=\{1,2,\ldots\}.
\end{align*}

The main advantage of PVAs is the existence of a closed and explicit formula
to compute the $\lambda$ bracket of any two elements of $\mathcal{A}$, in
terms of the bracket between the generators of $\mathcal{A}$. Such a formula
is called the \emph{master formula} and reads
\begin{equation}  \label{eq:masterformula}
\{f_\lambda g\}=\sum_{i,j=1}^n\sum_{l\geq 0}\sum_{m\geq 0}\frac{\partial g}{%
\partial u^j_{(m)}}\left(\lambda+\partial\right)^m\{u^i_{\lambda+%
\partial}u^j\}\left(-\lambda-\partial\right)^l\frac{\partial f}{\partial
u^i_{(l)}}
\end{equation}
where $n$ is the number of generators of $\mathcal{A}$ and $u^i_{(l)}$
denotes the $l$-th jet coordinate ($\partial u^i_{(l)}=u^i_{(l+1)}$).

Given a Hamiltonian operator $P^{ij}(\partial)$, the $\lambda$ bracket of
the corresponding PVA is obtained by setting the bracket between the
generators equal the transpose of the symbol of the operator, $\{u^i_\lambda
u^j\}=P^{ji}(\lambda)$. The strategy of our proof consists in obtaining the $%
\lambda$ bracket corresponding to the candidate Hamiltonian operator of
third order, and requiring that it must satisfy the skew-symmetry and
PVA-Jacobi properties. They are equivalent \cite{PV1} to the skew-symmetry
and the Jacobi identity for the Poisson bracket defined by the operator --
hence the conditions that we derive are the conditions for the operator to
be Hamiltonian.

\subsection{The $\protect\lambda$ bracket}

The operator $A$ defined by \eqref{eq:32} corresponds to the $\lambda$
bracket of the form
\begin{equation}  \label{eq:Adef}
\{u^i{}_{\lambda} u^j\}=(\lambda+\partial)\left(g^{ji}\lambda+c^{ji}_l
u^l_x+w^j_lu^l_x(\lambda+\partial)^{-1}w^i_mu^m_x\right)\lambda.
\end{equation}
For convenience, we express the $\lambda$ bracket \eqref{eq:Adef} in
potential coordinates $v^i_x=u^i$ where it takes the form $-\{u^i_\lambda
u^j\}=\{v^i_\lambda v^j\}' = \{v^i_\lambda v^j\}_L + \{v^i_\lambda v^j\}_N$,
with
\begin{equation}  \label{eq:20}
\{v^i_\lambda v^j\}_L= g^{ji}\lambda + c^{ji}_k v^k_{2x},\quad \{v^i_\lambda
v^j\}_N= w^j_lv^l_{2x}(\lambda+\partial)^{-1}w^i_mv^m_{2x},
\end{equation}
and all the functions depending on $v^i_x$ only. We choose to decompose the
bracket in its local part $\{v^i_\lambda v^j\}_L$ and its nonlocal part $%
\{v^i_\lambda v^j\}_N$. The nonlocal part is admissible, being a ratio of
local $\lambda$ brackets \cite{PV2}.

\subsection{The skew-symmetry condition}

The condition of skew-symmetry is equivalent to the conditions
\begin{equation}  \label{eq:21}
g^{ij} = g^{ji},\quad g^{ij}_{,k} = c^{ij}_k + c^{ji}_k.
\end{equation}
Indeed, for the local part we have
\begin{equation*}
\{v^i_\lambda
v^j\}_L=g^{ji}\lambda+c^{ji}_kv^k_{2x}=-{}_\to\{v^j_{-\lambda-\partial}v^i%
\}_L=g^{ij}\lambda+\partial_k g^{ij}v^k_{2x}-c^{ij}_kv^k_{2x}.
\end{equation*}
The two identities are the coefficients of $\lambda$ and of $v^k_{2x}$ in
both sides of the equation.

The nonlocal part is skew-symmetric by construction. This can be shown by
taking the formal series expansion of $(\lambda+\partial)^{-1}$,
substituting the powers of $\lambda$ with $(-\lambda-\partial)$ and taking a
double expansion. However this procedure can be replaced by a much simpler
one: in the nonlocal part of the bracket, $w^j_lv^l_{2x}(\lambda+%
\partial)^{-1}w^i_mv^m_{2x}$, the total derivative in the parenthesis acts
only on terms on its right-hand side. For shorthand, we can write $%
(\lambda+\partial)^{-1}$ as $(\lambda+\partial^{(i)})^{-1}$ where the
superscript means that it acts on $w^i_mv^m_{2x}$. By doing this, the actual
position of the operator $(\lambda+\partial)^{-1}$ becomes irrelevant. Thus,
\begin{equation*}
\{v^i_\lambda v^j\}_N=
\left(\lambda+\partial^{(i)}\right)^{-1}\left(w^j_lv^l_{2x}\right)%
\left(w^i_mv^m_{2x}\right).
\end{equation*}
On the other hand, the total derivative in the definition of skew-symmetry
acts on all the bracket, namely it can be interpreted as $%
\partial^{(i)}+\partial^{(j)}$. This means that
\begin{equation*}
{}_\to\{v^j_{-\lambda-\partial}v^i\}=\left(-\lambda-\partial^{(i)}-%
\partial^{(j)}+\partial^{(j)}\right)^{-1}\left(w^j_lv^l_{2x}\right)%
\left(w^i_mv^m_{2x}\right),
\end{equation*}
from which the skew-symmetry easily follows.

\subsection{The PVA-Jacobi identity}

The PVA-Jacobi identity for the bracket \eqref{eq:Adef} splits into four
parts, when taking into account local and nonlocal parts separately. Let us
adopt the shorthand notation
\begin{gather*}
T^{ijk}_{P,Q}(\lambda,\mu):=\{u^i_\lambda \{u^j_\mu u^k\}_P\}_Q, \\
J^{ijk}(A,B):=T^{ijk}_{A,B}(\lambda,\mu)-T^{jik}_{A,B}(\mu,%
\lambda)+T^{kij}_{A,B}(-\lambda-\mu-\partial,\lambda),
\end{gather*}
for the terms of the PVA-Jacobi identity, where $(P,Q)$ denote different $%
\lambda$ brackets and the last term in $J^{ijk}$ should be understood as $%
{}_\to\{u^k_{-\lambda-\mu-\partial}\{u^i_\lambda u^j\}_A\}_B$, using the
skew-symmetry property of the $\lambda$ bracket. The PVA-Jacobi identity can
then be written as
\begin{equation*}
J^{ijk}(A,A)=0.
\end{equation*}
By the linearity of the bracket, the PVA-Jacobi identity reads
\begin{equation*}
J^{ijk}(A,A)=J^{ijk}(L,L)+J^{ijk}(L,N)+J^{ijk}(N,L)+J^{ijk}(N,N)=0.
\end{equation*}
The computation of $J^{ijk}(L,L)$ is a straightforward application of the
master formula \eqref{eq:masterformula}.

The expressions involving nonlocal terms live in the space $V_{\lambda,\mu}$
whose homogeneous components have the basis $(\lambda^i\mu^{d-i},%
\lambda^{d+j}(\lambda+\mu)^{-j})$ with $i\in \mathbb{Z}$ and $j\in \mathbb{Z}%
_+$. We choose to isolate the nonlocal coefficients of the form $%
P[\mu][(\lambda+\partial)^{-1}w^i_sv^s_{2x}]$, $P[\lambda][(\mu+%
\partial)^{-1}w^j_sv^s_{2x}]$ and $w^k_sv^s_{2x}(\lambda+\mu+%
\partial)^{-1}P[\lambda]$ where $P[\nu]$ are polynomials in the formal
parameter $\nu$ with differential polynomials as coefficients. Expanding $%
(\nu+\partial)^{-1}f$ by $\sum_{k\geq0} (-1)^{k}\nu^{-1-k}\partial^kf$ in
the former two expressions we obtain elements in the subspace whose basis is
$\lambda^i\mu^{d-i}$; doing the same in the latter produces elements in the
subspace $(\lambda+\mu)^{-j}\lambda^{d+j}$. In fact, they give an infinite
number of coefficient, but it is apparent that the vanishing of the terms
for $k=0$ (corresponding to $\lambda^{-1}$, $\mu^{-1}$ and $%
(\lambda+\mu)^{-1}$, respectively) is necessary and sufficient for the
vanishing of all the expansion (because the further elements of the
expansion have a different total degree $d$, so correspond to other elements
of the basis).

To explicitly show how we compute the $\lambda$ bracket for nonlocal PVAs
and how we express the terms of the PVA-Jacobi identity in the basis of $%
V_{\lambda,\mu}$, we demonstrate the full computation for the term $%
T^{jik}_{N,L}(\mu,\lambda)$, i.e. the second summand of $J^{ijk}(N,L)$.
We have
\begin{equation*}
\{v^j{}_{\mu}\{v^i{}_\lambda v^k\}_N\}_L=\{v^j{}_\mu
w^k_mv^m_{2x}(\lambda+\partial)^{-1}w^i_nv^n_{2x}\}_L.
\end{equation*}
By Leibniz's property this expression equals
\begin{equation*}
[(\lambda+\partial)^{-1}w^i_nv^n_{2x}]\{v^j_{\mu}w^k_mv^m_{2x}%
\}_L+w^k_mv^m_{2x}\{v^j_{\mu}(\lambda+\partial)^{-1}w^i_nv^n_{2x}\}_L,
\end{equation*}
(the square brackets remind us that the pseudodifferential operator does not
act outside them). Using sesquilinearity on the second summand we obtain
\begin{align*}
&[(\lambda+\partial)^{-1}w^i_nv^n_{2x}]\{v^j_{\mu}w^k_mv^m_{2x}%
\}_L+w^k_mv^m_{2x}(\lambda+\mu+\partial)^{-1}\{v^j_{\mu}w^i_nv^n_{2x}\}_L
\notag \\
&=[(\lambda+\partial)^{-1}w^i_nv^n_{2x}]\left(w^k_{m,l}v^m_{2x}(\mu+%
\partial)(g^{lj}\mu+c^{lj}_sv^s_{2x})+w^k_l(\mu+\partial)^2(g^{lj}%
\mu+c^{lj}_sv^s_{2x})\right)  \notag \\
&+w^k_mv^m_{2x}(\lambda+\mu+\partial)^{-1}\left(w^i_{n,l}v^n_{2x}(\mu+%
\partial)(g^{lj}\mu+c^{lj}_sv^s_{2x})+w^i_l(\mu+\partial)^2(g^{lj}%
\mu+c^{lj}_sv^s_{2x})\right).  \notag
\end{align*}
The first line of the expression is of the form $[(\lambda+%
\partial)^{-1}w^i_nv^n_{2x}]P[\mu]$, so it is enough to expand $P$. On the
other hand, the second line of the expression is of the form $%
w^k_sv^s_{2x}(\lambda+\mu+\partial)^{-1}P[\mu]$ which cannot be immediately
expanded in the basis of $V_{\lambda,\mu}$. Consider for instance
\begin{equation*}
w^k_mv^m_{2x}(\lambda+\mu+\partial)^{-1}\left(w^i_{n,l}v^n_{2x}(\mu+%
\partial)(g^{lj}\mu+c^{lj}_sv^s_{2x})\right).
\end{equation*}
We proceed by replacing $(\mu+\partial)$ with $(\lambda+\mu+\partial)$, then
moving it to the right of its inverse to make them cancel out:
\begin{multline*}
w^k_mv^m_{2x}(\lambda+\mu+\partial)^{-1}\left(w^i_{n,l}v^n_{2x}(\mu+%
\partial)(g^{lj}\mu+c^{lj}_sv^s_{2x})\right)= \\
w^k_mv^m_{2x}(\lambda+\mu+\partial)^{-1}(\lambda+\mu+\partial)%
\left(w^i_{n,l}v^n_{2x}(g^{lj}\mu+c^{lj}_sv^s_{2x})\right) \\
-w^k_mv^m_{2x}(\lambda+\mu+\partial)^{-1}\left(\lambda
w^i_{n,l}v^n_{2x}(g^{lj}\mu+c^{lj}_sv^s_{2x})+%
\partial(w^i_{n,l}v^n_{2x})(g^{lj}\mu+c^{lj}_sv^s_{2x})\right) \\
=w^k_mv^m_{2x}w^i_{n,l}v^n_{2x}(g^{lj}\mu+c^{lj}_sv^s_{2x})+w^k_mv^m_{2x}(%
\lambda+\mu+\partial)^{-1}\left(-\lambda
w^i_{n,l}v^n_{2x}c^{lj}_sv^s_{2x}-%
\partial(w^i_{n,l}v^n_{2x})c^{lj}_sv^s_{2x}\right) \\
-w^k_mv^m_{2x}(\lambda+\mu+\partial)^{-1}\left(\mu\lambda
w^i_{n,l}v^n_{2x}g^{lj}+\mu\partial(w^i_{n,l}v^n_{2x})g^{lj}\right).
\end{multline*}
The same procedure is then repeated to eliminate $\mu$ in the last
parenthesis.

We proceed similarly to compute all the terms of the nonlocal part of
PVA-Jacobi identity. Finally, we get an overall expression for which we can
collects the coefficients of $\lambda^3$, $\mu^3$, $\lambda^2\mu$, $%
\lambda\mu^2$, $\lambda^2$, $\lambda\mu$, $\mu^2$, $\lambda$, $\mu$, $1$, $%
[(\lambda+\partial)^{-1}w^i_nv^n_{2x}]\mu^3$, $[(\lambda+%
\partial)^{-1}w^i_nv^n_{2x}]\mu^2$, $[(\lambda+\partial)^{-1}w^i_nv^n_{2x}]%
\mu$, $[(\lambda+\partial)^{-1}w^i_nv^n_{2x}]$, $[(\mu+%
\partial)^{-1}w^j_nv^n_{2x}]\lambda^3$, $[(\mu+\partial)^{-1}w^j_nv^n_{2x}]%
\lambda^2$, $[(\mu+\partial)^{-1}w^j_nv^n_{2x}]\lambda$, $%
[(\mu+\partial)^{-1}w^j_nv^n_{2x}]$, $w^k_sv^s_{2x}(\lambda+\mu+%
\partial)^{-1}\lambda^3$, $w^k_sv^s_{2x}(\lambda+\mu+\partial)^{-1}\lambda^2$%
, $w^k_sv^s_{2x}(\lambda+\mu+\partial)^{-1}\lambda$, $w^k_sv^s_{2x}(\lambda+%
\mu+\partial)^{-1}$. The last four terms, of course, mean that the
coefficients are differential polynomials on which $(\lambda+\mu+%
\partial)^{-1}$ acts.

Under the assumption of skew-symmetry of the brackets provided by relations %
\eqref{eq:21}, some of the coefficients in the previous expansion are
equivalent under the interchange of indices. We recall that the PVA-Jacobi
identity is fulfilled if and only if the aforementioned terms vanish for all
$(i,j,k)$; in particular, it is sufficient to consider the terms $%
\lambda^{-1}$, $\mu^{-1}$ and $(\lambda+\mu)^{-1}$ in the expansion of the
nonlocal part, respectively. The independent coefficients that need to be
set to 0 are hence the ones corresponding to $\lambda^3$, $\lambda^2\mu$, $%
\lambda^2$, $\lambda\mu$, $\lambda$, $\mu$, 1, $\lambda^{-1}\mu^3$, $%
\lambda^{-1}\mu^2$, $\lambda^{-1}\mu$, $\lambda^{-1}$. Moreover, each of
these coefficients can be further expanded in the jet variables $v^i_{2x}$, $%
v^i_{3x}$, $v^i_{4x}$ that appear in them, leading to a set of equations for
$g^{ij}$, $c^{ij}_k$ and $w^i_j$.

In particular, we have the following:

\begin{proposition}
Assuming the skew-symmetry of the bracket \eqref{eq:21}, the vanishing of
the following coefficients in the expansion of the PVA-Jacobi identity for
the $\lambda$ bracket \eqref{eq:Adef} is necessary and sufficient for the
vanishing of the whole expression.

\begin{enumerate}
\item Coefficient of $\lambda^3$: $g^{ip}c^{kj}_p+g^{kp}c^{ij}_p$.

\item Coefficient of $\lambda^2\mu$: $%
g^{kp}c^{ij}_p+g^{ip}c^{jk}_p-g^{jp}c^{ki}_p$.

\item Coefficient of $\lambda u^s_{3x}$: after some manipulations involving
the previous identities and their differential consequences we obtain $%
g^{kp}c^{ij}_{p,s}+c^{kj}_pg^{pi}_{,s}+c^{ki}_pc^{pj}_s-c^{ik}_pc^{pj}_s+g^{kp}w^i_pw^j_s.
$

\item Coefficient of $\lambda^{-1}\mu^3$: $-\left(g^{jp}w^k_p+g^{kp}w^j_p%
\right)$.

\item Coefficient of $\lambda^{-1}\mu^2v^s_{2x}$: after some manipulations
involving the previous identities and their differential consequences we
obtain $g^{kp}w^j_{p,s}+g^{jp}_{,s}w^k_p-c^{jk}_pwp_s+c^{kj}_pw^p_s$.
\end{enumerate}

All other coefficients in the expansion are either algebraic/differential
consequences of the former ones, or can be obtained from them by interchange
of indices and the skew-symmetry of the bracket.
\end{proposition}

This finishes the proof of Theorem \ref{Jacobi}.

\section{Concluding Remarks}

\begin{itemize}
\item We have demonstrated that nonlocal operators  (\ref{eq:32})
arise as Dirac reductions of local operators (\ref{eq:1}) to hyperplanes in
the flat coordinates. It remains to be proved that \textit{every} nonlocal
operator (\ref{eq:32}) can be obtained by this construction.

\item It makes sense to consider Hamiltonian operators with longer `nonlocal
tails' such as
\begin{equation}
A = \partial_x^{}\Big(g^{ij}\partial_x^{} + c^{ij}_k
u^k_x+\sum_{\alpha=1}^Nw^i_{\alpha k}u^k_x\partial_x^{-1}w^j_{\alpha l} u^l_x%
\Big)\partial_x^{},  \label{tail}
\end{equation}
which can be viewed as Dirac reductions of local operators (\ref{eq:1}) to
(special!) linear subspaces of codimension $N$.  Thus, it was observed recently that the oriented associativity
equations (which can be reduced to 6-component systems of hydrodynamic type) possess, in addition to a local first-order Hamiltonian structure
\cite{ps}, a third-order Hamiltonian operator with nonlocal tail of length $N=3$
\cite{Vit}.

 Taking  the 5th equation (\ref{cond1}) for $w$, 
\begin{equation}
w_{ij,l}=c_{ij}^{s}w_{sl},
\end{equation}%
and calculating its compatibility conditions  $(w_{ij,l})_{k}=(w_{ij,k})_{l}$ we obtain a linear system for $w$:
\begin{equation}
(c_{ij,k}^{p}-c_{ij}^{q}c_{qk}^{p})w_{ps}=(c_{ij,s}^{p}-c_{ij}^{q}c_{qs}^{p})w_{pk},
\label{b}
\end{equation}%
recall that the coefficients $c^p_{ij}$ are uniquely determined by the Monge metric $g_{ij}$.
The analysis of this system for different Segre types of Monge metrics suggests that for $n=3$ the length $N$ of the nonlocal tail cannot exceed 1. This implies that 
Theorem \ref{3D} of Section \ref{sec:3} gives a complete list of 3-component nonlocal operators of type (\ref{tail}). As an immediate corollary we obtain that already for $n=3$ {\it not every} Monge metric generates a nonlocal Hamiltonian operator of type (\ref{tail}) (note that for $n=2$ every Monge metric gives rise to a nonlocal operator with tail of length $N\leq 1$, see Section \ref{sec:2}).
\end{itemize}

\section{Acknowledgements}

Matteo Casati was supported by the INdAM-Cofund-2012 Marie Curie fellowship
`MPoisCoho'. Maxim Pavlov was partially supported by the grant of Presidium of
RAS "Fundamental Problems of Nonlinear Dynamics" and by the RFBR grant
18-01-00411-a. Raffaele Vitolo recognises financial support from the
Loughborough University's Institute of Advanced Studies, LMS scheme 2 grant,
INFN by IS-CSN4 \emph{Mathematical Methods of Nonlinear Physics}, GNFM of
Istituto Nazionale di Alta Matematica and Dipartimento di Matematica e Fisica
``E. De Giorgi'' of the Universit\`a del Salento. We thank CIRM (Trento) for
their kind hospitality via the `research in pairs' programme where this project
was completed.


\begin{thebibliography}{99}

\bibitem{BP} \emph{A.V. Balandin, G.V. Potemin}, \newblock On non-degenerate
differential-geometric Poisson brackets of third order, Russian Mathematical
Surveys \textbf{56}, no. 5 (2001) 976-977.

\bibitem{PV1} \emph{A. Barakat, A. De Sole, V.G. Kac}, Poisson vertex
algebras in the theory of Hamiltonian equations, Jpn. J. Math. \textbf{4},
no. 2 (2009), no. 2, 141-252.

\bibitem{PV2} \emph{A. De Sole, V.G. Kac}, Non-local Poisson structures and
applications to the theory of integrable systems, Jpn. J. Math. \textbf{8},
no. 2 (2013) 233-347.

\bibitem{Dolgachev} \emph{\ I.V. Dolgachev}, Classical algebraic geometry. A
modern view, Cambridge University Press, Cambridge, 2012, 639 pp.

\bibitem{Doyle} \emph{P.W. Doyle}, \newblock Differential geometric Poisson
bivectors in one space variable, J. Math. Phys. \textbf{34}, no. 4 (1993)
1314-1338.

\bibitem{DN2} \emph{B.A. Dubrovin and S.P. Novikov}, Poisson brackets of
hydrodynamic type, Soviet Math. Dokl. \textbf{30}, no. 3 (1984) 651-2654.

\bibitem{Dub} \emph{B.A. Dubrovin}, Geometry of 2D topological field
theories, Lecture Notes in Mathematics, V.1620, Berlin, Springer, 120-348.

\bibitem{Fer95} \emph{E.V. Ferapontov}, Nonlocal Hamiltonian operators of
hydrodynamic type: differential geometry and applications, Amer. Math. Soc.
Transl. (2) \textbf{170} (1995) 33-58.

\bibitem{FN} \emph{E.V. Ferapontov, C.A.P. Galvao, O. Mokhov, Y. Nutku}, %
\newblock Bi-Hamiltonian structure of equations of associativity in 2-d
topological field theory, Comm. Math. Phys. \textbf{186 }(1997) 649-669.

\bibitem{fmoss} \emph{E.V. Ferapontov, J. Moss}, Linearly degenerate PDEs
and quadratic line complexes, Comm. Anal. Geom. \textbf{23}, no.1 (2015)
91-127.

\bibitem{FPV} \emph{E.V. Ferapontov, M. V. Pavlov, R.F. Vitolo},
Projective-geometric aspects of homogeneous third-order Hamiltonian
operators, J. Geom. Phys. \textbf{85} (2014) 16-28.

\bibitem{FPV1} \emph{E.V. Ferapontov, M.V. Pavlov, R.F. Vitolo}, Towards the
classification of homogeneous third-order Hamiltonian operators, IMRN no.
\textbf{22} (2016) pp. 6829-6855; doi:10.1093/imrn/rnv369.

\bibitem{FPV2} \emph{Ferapontov, M.V. Pavlov, R.F. Vitolo}, Systems of
conservation laws with third-order Hamiltonian structures, Lett. Math. Phys.
(2017), DOI: 10.1007/s11005-018-1054-3; arXiv:1703.06173.

\bibitem{GD} I.M. Gelfand, I. Ja. Dorfman, Hamiltonian operators and
algebraic structures associated with them, Funktsional. Anal. i Prilozhen.
\textbf{13}, no. 4 (1979) 13-30, 96.

\bibitem{Jess} \emph{C.M. Jessop}, A treatise on the line complex, Cambridge
University Press 1903.

\bibitem{KN1} \emph{J. Kalayci, Y. Nutku}, Bi-Hamiltonian structure of a
WDVV equation in 2d topological field theory, Phys.\ Lett.\ A \textbf{227}
(1997), 177--182.

\bibitem{KN2} \emph{J. Kalayci, Y. Nutku}, Alternative bi-Hamiltonian
structures for WDVV equations of associativity, J. Phys.\ A: Math.\ Gen.\
\textbf{31} (1998) 723-734.

\bibitem{MMZ} \emph{W.-X. Ma, S. Manukure, H.-C. Zheng,} A counterpart of
the Wadati-Konno-Ichikawa soliton hierarchy associated with ${so}(3,R)$, Z.
Naturforsch. \textbf{69}a (2014) 411-419; %
\url{http://arxiv.org/abs/1405.1089}.

\bibitem{OM98} \emph{O.I. Mokhov}, \newblock Symplectic and Poisson
structures on loop spaces of smooth manifolds, and integrable systems,
Russian Math. Surveys \textbf{53}, no. 3 (1998) 515-622.

\bibitem{pv} \emph{M.V. Pavlov, R. Vitolo,} \newblock On the bi-Hamiltonian
Geometry of WDVV Equations, Lett. Math. Phys. \textbf{105}, no. 8 (2015)
1135-1163.

\bibitem{Vit} \emph{M.V. Pavlov, R.F. Vitolo}, On the bi-Hamiltonian
Geometry of Oriented WDVV Equations, to appear.

\bibitem{ps} \emph{M.V. Pavlov, A. Sergyeyev}, \newblock Oriented
associativity equations and symmetry consistent conjugate curvilinear
coordinate nets, J. Geom. Phys. \textbf{85} (2014), 46-59.

\bibitem{GP87} \emph{G.V. Potemin}, \newblock On Poisson brackets of
differential-geometric type, Soviet Math. Dokl. \textbf{33} (1986) 30-33.

\bibitem{GP97} \emph{G.V. Potemin}, \newblock On third-order Poisson
brackets of differential geometry, Russ. Math. Surv. \textbf{52} (1997)
617-618.

\bibitem{GP91} \emph{G.V. Potemin}, \newblock Some aspects of differential
geometry and algebraic geometry in the theory of solitons. PhD Thesis,
Moscow, Moscow State University (1991) 99 pages.

\bibitem{WKI} \emph{M. Wadati, K. Konno, Y. H. Ichikawa}, New integrable
nonlinear evolution equations, J. Phys. Soc. Japan \textbf{47}, no. 5 (1979)
1698-1700.
\end{thebibliography}
\end{document}